\documentclass{article}

\usepackage[a4paper]{geometry}

\title{Quantum Synchronizable Codes From Cyclotomic Classes of Order Two over $\mathbb{Z}_{2q}$}
\author{Tao Wang$^{*}$, Tongjiang Yan, Vladimir Sidorenko, Xueting Wang\\
\small China University of Petroleum\\[-0.8ex]
\small Qingdao, China\\
\small\tt $^*$1409010215@s.upc.edu.com\\
}

\date{June 2021}

\usepackage{url}
\usepackage{graphicx}
\usepackage{charter}
\usepackage[charter]{mathdesign}
\usepackage{algorithm}
\usepackage{amsthm}
\usepackage{amsmath}
\newtheorem{theorem}{Theorem}
\newtheorem{lemma}{Lemma}
\newtheorem{example}{Example}
\newtheorem{definition}{Definition}
\usepackage{algorithmicx}
\usepackage{algpseudocode}
\usepackage{booktabs}
\usepackage{subfigure}
\usepackage{caption}

\usepackage{multirow}
\usepackage{algorithm,algpseudocode,float}
\usepackage{lipsum}


\begin{document}
\bibliographystyle{plain}
\maketitle
\begin{abstract}
Quantum synchronizable codes are kinds of quantum error-correcting codes that can not only correct the effects of quantum noise on qubits but also the misalignment in block synchronization. This paper contributes to constructing two classes of quantum synchronizable codes by the cyclotomic classes of order two over $\mathbb{Z}_{2q}$, whose synchronization capabilities can reach the upper bound under certain conditions. Moreover, the quantum synchronizable codes possess good error-correcting capability towards bit errors and phase errors.

\end{abstract}
\section{Introduction}
\label{intro}
\hspace{1.5em}In recent decades, quantum information and quantum computation have experienced remarkable progress. Quantum error-correcting codes are fundamental tool in quantum communication that allows for quantum information processing in a noisy environment \cite{2002QuantumC,1998Good,2020Quantum}. Typically, quantum error-correcting codes are designed to correct the effects of bit errors and phase errors caused by Pauil operator, which correspond to additive noise in classical encoding theory. Misalignment is another type of errors that can also cause catastrophic failure in quantum information transmission. Misalignment can be brought in wherever the information processing device misidentifies the boundaries of an information stream. For instance, suppose that the quantum information which we transmit can be expressed by an ordered sequence of information blocks. Each chunk of information is encoded into a block of three consecutive qubits in a stream of qubits $|q_i\rangle$. If three chunks of information are encoded, we have $9$ ordered qubits $\left(|q_0\rangle|q_1\rangle|q_2\rangle|q_3\rangle|q_4\rangle|q_5\rangle|q_6\rangle|q_7\rangle|q_8\rangle\right)$ where each of the three blocks $\left(|q_0\rangle|q_1\rangle|q_2\rangle\right)$, $\left(|q_3\rangle|q_4\rangle|q_5\rangle\right)$ and $\left(|q_6\rangle|q_7\rangle|q_8\rangle\right)$ forms an information chunk. We assume that the synchronization was established at the beginning of the transmission, but synchronization may be lost during the quantum communication or quantum computation. This type of error occurs when the receiver incorrectly locates the boundary of each block of data by a certain number of positions towards the left or right. For example, the receiver wrongly read out $\left(|q_5\rangle|q_6\rangle|q_7\rangle\right)$ instead of the correct block $\left(|q_6\rangle|q_7\rangle|q_8\rangle\right)$. For more details, see \cite{BOSE1967616,2012Block}.

As a subclass of quantum error-correcting codes, the mainly function of quantum synchronizable codes is to correct the errors caused by quantum noise and block synchronization in the process of quantum information transmission. In order to ensure information security, it is of great significance to study the construction of quantum synchronizable codes \cite{2020Quantum}. In 2013, Fujiwara proposed the framework of quantum block synchronization and gave the first example of quantum synchronizable codes with standard quantum error correction capability and a general construction framework \cite{fujiwara2013algebraic,fujiwara2014quantum}.

In 2014, Xie proposed to use quadratic residual codes and hypercodes to construct quantum synchronizable codes \cite{2014QuantumXie}. The simple construction of the quantum synchronizable codes given by him allows  their synchronization ability to reach the upper limit. After that, in 2019 \cite{2019QuantumLi}, Li constructed two classes of quantum synchronizable codes from cyclotomic classes of order four. It also has been proved that the codes have good error correction performance.  Quantum synchronizable codes were further studied in \cite{XiaLi2021A,2018Non,2016q}.

In this work, we propose a new method to constructing two classes of quantum synchronizable codes. In certain conditions, their misalignment tolerance can reach the limit. Moreover, the quantum synchronizable codes possess good error-correcting capability towards bit errors and phase errors.

\section{Preliminaries}
\label{section1}
\subsection{Cyclic codes and their duals}
\hspace{1.5em}Let $\mathbb{F}_r$ be a finite field with $r$ elements. An $[n,k,d]_r$ linear code $C$ is a $k$-dimensional subspace of the $n$-dimensional vector space $\mathbb{F}_r^n$ with $\mathrm{min}\{\mathrm{wt}(v)|v\in C,v\neq 0 \}=d$, where the weight $\mathrm{wt}(v)$ is the number of nonzero coordinates of $v$. $C$ is called optimal, if its distance reaches the Hamming bound, see e.g., \cite{Brouwer98}. $C$ is called almost optimal if the code with parameters $[n,k,d+1]_r$ reaches Hamming bound. $C$ is a cyclic code if any codeword $c=(c_0,c_1,...,c_{n-1})\in C$ implies $(c_{n-1},c_{0},...,c_{n-2})\in C$. If $C$ can be identified with
\begin{equation*}
	I(C)=\{c_0+c_1x+c_2x^2+...+c_{n-1}x^{n-1}|(c_0,c_1,c_2,...,c_{n-1})\in C\},
\end{equation*}
then $I(C)$ can be regarded as an ideal of the principal ideal ring $ \mathbb{F}_r[x]/{(x^n-1)} $ and generated by the monic polynomial $g(x)$. The Euclidean dual code of $C$ is defined as
\begin{equation*}
	C^\perp=\{x\in \mathbb{F}_r^n|(x,c)=0 , \forall c \in C\},
\end{equation*}
where $(x,c)$ is the Euclidean inner product of $x$ and $c$. $C^\perp$ is also a cyclic code, and its generator polynomial is given by
\begin{equation}\label{Cperp}
	\tilde{h}(x)=h(0)^{-1}x^kh\left(x^{-1}\right).
\end{equation}
where $h(x)=\frac{x^n-1}{g(x)}$ is the parity check polynomial of $C$ \cite{2003Fundamentals}, and $\tilde{h}(x)$ is called the reciprocal polynomial of $h(x)$.

Let $ C'$  be another cyclic code in $\mathbb{F}_r^n$. If $C\subseteq C' $, $C$ is said to be $C'$-containing. Then the generator polynomial of $C'$ can divide any codeword of $C$. If $ C^\perp \subset C $, $C$ is called dual-containing \cite{2019QuantumLi}.

\subsection{Cyclotomic class}
\hspace{1.5em}Let $q=2f+1$ be an odd prime and $g$ be a fixed primitive root of $q$. Define $D_0^{(q)}=\{g^{2s}|s=0,1,...,\frac{q-1}{2}-1\}$ and $D_1^{(q)}=\{g^{2s+1}|s=0,1,...,\frac{q-1}{2}-1\}$ as the quadratic residue set and the quadratic non-residue set respectively.
\begin{definition}
	Let $q$ be an odd prime. Suppose $g$ is a common primitive root modulo both $q$ and $2q$ \cite{2005Elementary}. Then $\mathrm{ord}_{2q}(g)=q-1$, and the cyclotomic classes of order two over $\mathbb{Z}_{2q}$ are
	\begin{equation*}
		\begin{aligned}
			&D_0^{(2q)}=\{g^{2s}  | s=0,1,...,\frac{q-1}{2}-1\},\;D_1^{(2q)}=\{g^{2s+1} | s=0,1,...,\frac{q-1}{2}-1\}, \\
			&E_0^{(2q)}=2D_0^{(q)}=\{2u |u\in D_0^{(q)}\},\quad E_1^{(2q)}=2D_1^{(q)}=\{2u  | u \in D_1^{(q)}\}. \\
		\end{aligned}
	\end{equation*}
\end{definition}

Denote $\mathbb{Z}_q^*=\mathbb{Z}_q\backslash \{0\}$ and $\mathbb{Z}_{2q}^*=\mathbb{Z}_{2q}\backslash \{0\}$. Then we have
\begin{equation*}
	\mathbb{Z}_{q}^*=D_0^{(q)} \cup D_1^{(q)},\; \mathbb{Z}_{2q}^*=\{q\}\cup D_0^{(2q)} \cup D_1^{(2q)} \cup E_0^{(2q)}\cup E_1^{(2q)}.
\end{equation*}
Since all the $q-1$ elements in $E_0^{(2q)}\cup E_1^{(2q)}$ are even, the $q-1$ elements in $D_0^{(2q)} \cup D_1^{(2q)}$ must be odd. Then $-1\in D_0^{(2q)} \cup D_1^{(2q)}$.
\begin{lemma}\label{lemma1}\cite{2005Elementary}
	Let $q$ be an odd prime.
	\begin{equation*}
		\begin{aligned}
			&(a) \quad \mbox{If } q\equiv 1 \bmod 4,\mbox{ then } -1 \in D_0^{(2q)}. \\
			&(b) \quad \mbox{If } q\equiv 3 \bmod 4,\mbox{ then } -1 \in D_1^{(2q)}.
		\end{aligned}
	\end{equation*}
\end{lemma}
\begin{lemma}\label{lemma2}\cite{2016Linear}
	Let $q$ be an odd prime and $i,j \in \{0,1\}$. If $v \in D_i^{(2q)}$, we have
	\begin{equation*}
		(a)\;vD_j^{(2q)}=D_{(i+j)\bmod 2}^{(2q)}, \quad (b)\;vE_{j}^{(2q)}=E_{(i+j)\bmod 2}^{(2q)}.
	\end{equation*}
\end{lemma}

\subsection{Cyclotomy coset}
\hspace{1.5em}Let $r\in D_0^{(2q)}$ be a prime, and $\mathrm{gcd}(r,2q)=1$. Then the $r$-cyclotomy coset modulo $2q$ containing the integer $t$ is defined as
\begin{equation}\label{e9}
	C_{(t,2q)}=\{tr^i \bmod 2q|i \in \mathbb{N}\},
\end{equation}
where $\mathbb{N}$ is the set of all nonnegative integers.

\begin{lemma}\label{lemma3}
	Let $r=g^{2k} $ be an odd prime, and $\mathrm{ord}_{2q}(r)=\ell_{k}$. If $\ell_{m} \leq \ell_{k}$ is an integer that satisfies $tr^{\ell_{m}} \equiv t \bmod 2q$, then the size of $C_{(t,2q)}$ is given by
	\begin{equation*}
		|C_{(t,2q)}|=
		\begin{cases}
			\ell_{m} =1, & \mbox{ if }\mathrm{gcd}(2q,t) = q, \\
			\ell_{m} =\ell_{k}, & \mbox{ otherswise}.
		\end{cases}
	\end{equation*}
\end{lemma}
\begin{proof}
	Since $t<2q$ is an integer, there are three potential values of $\mathrm{gcd}(2q,t)$. First, let $\mathrm{gcd}(2q,t)=1$. The equation $tr^{\ell_{m}} \equiv t \bmod 2q$ means $r^{\ell_{m}} \equiv 1\bmod 2q$. Then we have $\ell_{m}=\ell_{k}$. Second, let $\mathrm{gcd}(2q,t)=2$. By $tr^{\ell_{m}}\equiv t \bmod 2q$, then $q|\frac{t}{2}(r^{\ell_{m}}-1)$. Since $q$ is an odd prime, we have $\mathrm{gcd}(q,\frac{t}{2})=1$ and $q|(r^{\ell_{m}}-1)$. Note that $r$ is an odd prime. Then $\ell_{m}=\ell_{k}$. Finally, let $\mathrm{gcd}(2q,t)=q$. Then $2|\frac{t}{q}(r^{\ell_{m}}-1)$. For $C_{(t,2q)}=\{t,tr,\cdots,tr^{\ell_m-1}\}$, it is clear that $tr = t \bmod 2q$. Then $\ell_m = 1$. We have finish this proof.
\end{proof}

From now on, let $\eta$ be a $2q$-$th$ primitive root of unity over some extension field of $\mathbb{F}_{r}$. Then, it is known that the polynomials
\begin{equation}\label{e6}
	\begin{aligned}
		&g_{D_i}^{(2q)}(x) = \prod_{j\in D_{i}^{(2q)}}(x-\eta^{j})\in \mathbb{F}_r[x],\\
		&g_{E_i}^{(2q)}(x) = \prod_{j\in E_{i}^{(2q)}}(x-\eta^{j})\in \mathbb{F}_r[x],
	\end{aligned}
\end{equation}
for any $i\in \{0,1\}$.

\section{Cyclic codes from generalized cyclotomy of order two}
\label{section2}
\hspace{1,5em}In this section, we first obtain some dual-containing cyclic codes by the cyclotomic classes of order two over $\mathbb{Z}_{2q}$ and discuss their Hamming distances. Furthermore, we introduce the augmented cyclic codes of the dual-containing codes obtained.

\subsection{Dual containing codes}
\hspace{1.5em}Since $\eta^{q}=-1$. Thus, we can get the factorization of $x^{2q}-1$ as
\begin{equation}\label{e10}
	x^{2q}-1=\prod_{j=0}^{2q-1}(x-\eta^j)=(x+1)(x-1)\prod_{i=0}^{1}g_{D_i}^{(2q)}(x)g_{E_i}^{(2q)}(x).
\end{equation}

Let $C_{D_i}^{(2q)}$, $C_{E_i}^{(2q)}$ and $C_{{D_iE_j}}^{(2q)}$ be the cyclic codes generated by $g_{D_i}^{(2q)}(x)$, $g_{E_i}^{(2q)}$ and $g_{D_i}^{(2q)}(x)g_{E_j}^{(2q)}(x)$ respectively, where $i,j\in \{0,1\}$. Then we have the following properties.
\begin{lemma}\label{lemma6}
	Let $q=4m+3$ be an odd prime. By using the notation above, we have
	\begin{equation*}
		(a)\;{C_{D_i}^{(2q)}}^{\perp}\subset C_{D_i}^{(2q)},\quad (b)\;{C_{E_i}^{(2q)}}^{\perp}\subset C_{E_i}^{(2q)},\quad (c)\;{C_{{D_iE_j}}^{(2q)}}^{\perp}\subset C_{{D_iE_j}}^{(2q)}.
	\end{equation*}
	
\end{lemma}
\begin{proof}
	By (\ref{Cperp}), the reciprocal polynomial of $g_{D_i}^{(2q)}(x)$ is
	\begin{equation*}
		\tilde{g}_{D_i}^{(2q)}(x)=\left(g_{D_i}^{(2q)}(0)\right)^{-1}x^{\mathrm{deg}\left(g_{D_i}^{(2q)}(x)\right)}g_{D_i}^{(2q)}\left(x^{-1}\right),
	\end{equation*}
	Namely,
	\begin{equation*}
		g_{D_i}^{(2q)}(x)=(x-\eta^{u_{i_1}})(x-\eta^{u_{i_2}})\cdots(x-\eta^{u_{i_{\frac{e}{2}}}}),
	\end{equation*}
	where $D_i^{(2q)}=\{u_{i_\delta}|1\leq \delta \leq \frac{e}{2}\}$. Then
	\begin{equation*}
		\begin{split}
			\tilde{g}_{D_i}^{(2q)}(x)=&(-\eta^{u_{i_1}})^{-1}(-\eta^{u_{i_2}})^{-1}\cdots(-\eta^{u_{i_{\frac{e}{2}}}})^{-1}x^{\frac{e}{2}} \\
			&(x^{-1}-\eta^{u_{i_1}})(x^{-1}-\eta^{u_{i_2}})\cdots (x^{-1}-\eta^{u_{i_{\frac{e}{2}}}}) \\
			=&(x-\eta^{-u_{i_1}})(x-\eta^{-u_{i_2}})\cdots (x-\eta^{-u_{i_{\frac{e}{2}}}}).
		\end{split}
	\end{equation*}
	According to Lemmas \ref{lemma1} and \ref{lemma2}, it was already seen that $-D_i^{(2q)}=D_{(i+1)}^{(2q)}$, then
	\begin{equation}\label{e7}
		\tilde{g}_{D_i}^{(2q)}(x)=g_{D_{(i+1)}}^{(2q)}(x),
	\end{equation}
	for any $i=\{0,1\}$.
	
	Moreover, since $-E_i^{(2q)}=E_{(i+1)}^{(2q)}$, it is obvious that $\left(-D_i^{(2q)}\right) \cup \left(-E_j^{(2q)}\right)=D_{(i+1)}^{(2q)} \cup E_{(j+1)}^{(2q)}$. Then
	\begin{equation}\label{e8}
		\tilde{g}_{E_i}^{(2q)}(x)=g_{E_{(i+1)}}^{(2q)},\quad \tilde{g}_{D_i}^{(2q)}(x)\tilde{g}_{E_j}^{(2q)}(x)=g_{D_{(i+1)}}^{(2q)}(x)g_{E_{(j+1)
		}}^{(2q)}(x),
	\end{equation}
	Thus the parity check polynomial of $C_{D_i}^{(2q)}$ is
	\begin{equation*}
		h(x) = \frac{x^{2q}-1}{g_{D_i}^{(2q)}(x)}=(x-1)(x+1)g_{D_{(i+1)}}^{(2q)}(x)\prod _{i=0}^{1}g_{E_i}^{(2q)}(x).
	\end{equation*}
	And by $(\ref{Cperp})$ and $(\ref{e7})$,
	\begin{equation*}
		\tilde{h}(x)=(x-1)(x+1)g_{D_i}^{(2q)}(x)\prod _{i=0}^{1}g_{E_i}^{(2q)}(x).
	\end{equation*}
	Then ${C_{D_i}^{(2q)}}^{\perp}=\langle\tilde{h}(x)\rangle$. Since $g_{D_i}^{(2q)}(x) | \tilde{h}(x)$, we obtain ${C_{D_i}^{(2q)}}^{\perp} \subset C_{D_i}^{(2q)}$. The other conclusions can be obtained trivially.
\end{proof}

\begin{example}
Let $q=4f+3$ be an odd prime, and $r\in D_0^{(2q)}$. Table \ref{biaolall} gives some examples of cyclic codes and their dual codes. All computations have been done by MAGMA \cite{Magma}.
\begin{table}[h]
	\centering
	\caption{Dual-containing cyclic codes from cyclotomic classes}  %
	\label{biaolall}
	\begin{tabular}{lll}
		\hline\noalign{\smallskip}
		Codes   &  Duals   &  Optimal or almost optimal \cite{Grassl:codetables} \\
		\noalign{\smallskip}\hline\noalign{\smallskip}
		$ C_{D_{0}}^{(22)}=[22,17,2]_3 $&$ {C_{D_{0}}^{(22)}}^{\perp}=[22,5,12]_3 $& $C_{D_{0}}^{(22)}$almost optimal, ${C_{D_{0}}^{(22)}}^{\perp}$optimal \cite{Grassl:codetables}  \\
		$ C_{E_{0}}^{(22)}=[22,17,2]_3 $&$ {C_{E_{0}}^{(22)}}^{\perp}=[22,5,12]_3 $& $C_{E_{0}}^{(22)}$almost optimal, ${C_{E_{0}}^{(22)}}^{\perp}$optimal \cite{Grassl:codetables}  \\
		$C_{D_0E_0}^{(22)}=[22,12,7]_3$ & ${C_{D_0E_0}^{(22)}}^{\perp}=[22,10,9]_3$ & both optimal \cite{Grassl:codetables} \\
		\noalign{\smallskip}\hline
	\end{tabular}
\end{table}
\end{example}

\subsection{Augmented cyclic codes}
\hspace{1.5em}A cyclic code $C'=\langle g'(x) \rangle$ is called augmented cyclic code of $C=\langle g(x) \rangle$ if $C \subset C'$. It follows that $g'(x)|g(x)$. By (\ref{e9}), the unique irreducible minimal polynomial of $\eta^t$ over $\mathbb{F}_{r}[x]$ is $M_t(x)=\prod_{j\in  C_{(t,2q)}}(x-\eta^j)$.

\begin{lemma}\label{theorem1}\cite{2003FiniteField}
	Consider the cyclic code $C_{D_i}^{(2q)}=\langle g_{D_i}^{(2q)}(x) \rangle$. If the size of $C_{(1,2q)}$ is $\ell_{k}$, then the generator polynomial $g_{D_i}^{(2q)}(x)$ can be expressed as the product of $\theta = \frac{q-1}{2\ell_{k}}$ irreducible polynomials of degree $\ell_{k}$ over $\mathbb{F}_{r}[x]$ as follows
	\begin{equation}\label{e11}
		g_{D_i}^{(2q)}(x)=\prod_{j=1}^{\theta} M_{d_{i_j}}(x),
	\end{equation}
	where $d_{i_j}$ are some numbers such that $D_i^{(2q)} = \cup_{j=1}^{\theta} C_{(d_{i_j},2q)}$, and $\mathrm{deg}(M_{d_{i_j}})=\ell_{k}$, for all $d_{i_j}$.
\end{lemma}
\begin{proof}
	Since $q$ is an odd prime, by Lemma \ref{lemma3}, it is clear that $|C_{(t,2q)}|=\ell_{k}$ for any integer $t\in \{1,2,...,q-1,q+1,...,2q-1\}$. Since $r \in D_0^{(2q)}$, it is known that $D_i^{(2q)}=\cup_{j=1}^{\theta}C_{(d_{i_j},2q)}$, where $d_{i_1},d_{i_2},...,d_{i_{\theta}}\in \{1,2,...,q-1,q+1,...,2q-1\}$ are some integers. Moreover, we have $\theta = \frac{|D_i^{(2q)}|}{|C_{(t,2q)}|}=\frac{q-1}{2\ell_{k}}$.
\end{proof}
\begin{example}
	Consider the cyclotomic classes of order two modulo $2q$, where $q=13$.
	\begin{equation*}
		\begin{aligned}
			&D_0^{(26)}=\{1,23,9,25,3,17\},\quad\;\, D_1^{(26)}=\{7,5,11,19,21,15\}, \\
			&E_0^{(26)}=\{2,20,18,24,6,8\},\quad E_1^{(26)}=\{14,10,22,12,16,4\}, \quad \{13\}. \\
		\end{aligned}
	\end{equation*}
	
	Let $r_1=3\in D_0^{(2q)}$. Then $\mathrm{ord}_{26}(r_1)=3$ and by $(\ref{e9})$, the $r_1-$ cyclotomy cosets modulo $26$ are
	\begin{equation*}
		\begin{aligned}
			&C_{(1,26)}=\{1,3,9\},\quad\quad C_{(2,26)}=\{2,6,18\}, \quad \;\; C_{(4,26)}=\{4,12,10\}, \\
			&C_{(5,26)}=\{5,15,19\},\quad C_{(7,26)}=\{7,21,11\},\quad C_{(8,26)}=\{8,24,20\} \\
			&C_{(13,26)}=\{13\},\qquad\quad C_{(14,26)}=\{14,16,22\}, C_{(17,26)}=\{17,25,23\}.
		\end{aligned}
	\end{equation*}
	Thus $D_0^{(26)}=C_{(1,26)}\cup C_{(17,26)}$, which is equivalent to $g_{D_0}^{(26)}=M_1(x)M_{17}(x)$. And in this way, the following equations can be drawn:
	\begin{equation*}
		D_1^{(26)}\!=\!C_{(7,26)}\!\cup\! C_{(5,26)},E_0^{(26)}\!=\!C_{(2,26)}\!\cup\! C_{(8,26)},E_1^{(26)}\!=\!C_{(14,26)}\!\cup\! C_{(4,26)}.
	\end{equation*}
	
	It is easy to deduce that the augmented cyclic code of $C_{D_{0}}^{(26)}$ is $C_{D_0}'=\langle M_{t}(x)\rangle$, where $t\in \{1,17\}$. Similar for other situations. Furthermore, since the augmented code obtained in this way is also dual-containing, we have the following lemma.
\end{example}

\begin{lemma}\label{lemma7}
	Let $q=4m+3$ be an odd prime, where $m\in \mathbb{N}$. Let $C_{D_i}^{(2q)}$, $C_{E_i}^{(2q)}$ and $C_{D_iE_j}^{(2q)}$ be the dual-containing codes defined in Lemma $\ref{lemma6}$. Then the following conclusions hold.
	
	$(a)$ If $C_{D_i}'=\langle \frac{g_{D_i}^{(2q)}(x)}{\prod_{i\in A}M_i(x)} \rangle$, then $C_{D_i}^{(2q)} \subset C_{D_i}'$, where $A$ is some proper subset of $\{d_{i_1},d_{i_2},...,d_{i_\theta}\}$.
	
	$(b)$ If $C_{E_i}'=\langle \frac{g_{E_i}^{(2q)}(x)}{\prod_{i\in B}M_i(x)} \rangle$, then $C_{E_i}^{(2q)} \subset C_{E_i}'$, where $B$ is some proper subset of $\{e_{i_1},e_{i_2},...,e_{i_\theta}\}$.
	
	$(c)$ If $C_{D_iE_j}'=\langle \frac{g_{D_i}^{(2q)}(x)g_{E_j}^{(2q)}(x)}{\prod_{i\in S} M_i(x)}\rangle$, then $C_{D_iE_j}^{(2q)} \subset C_{D_iE_j}'$, where $S$ is some proper subset of $\{d_{i_1},d_{i_2},...,d_{i_\theta} \}\cup \{e_{j_1},e_{j_2},...,e_{j_\theta} \}$.
\end{lemma}
\begin{proof}
	The results can be obtained from the definition of dual-containing codes.
\end{proof}

\section{Quantum synchronizable codes from the chain cyclic codes}
\hspace{1,5em}We first give some general facts of quantum synchronizable codes. A $(c_l,c_r)$-$[[n,k]]$ quantum synchronizable code is a quantum error-correcting code that corrects not only Pauli errors, but also block misalignment to the left by $c_l$ qubits and to the right by $c_r$ qubits for some nonnegative integers $c_l$ and $c_r$. Then the general construction framework of $r$-ary quantum synchronizable codes is given as following.
\begin{lemma}\cite{2018Non}\label{theorem2}
	Let $C_1=\langle g_1(x)\rangle$ and $C_2=\langle g_2(x)\rangle$ be two cyclic codes of parameters $[n,k_1,d_1]_r$ and $[n,k_2,d_2]_r$ respectively over $\mathbb{F}_r$ with $k_1>k_2$ such that $C_2\subset C_1$ and $C_2^{\perp}\subseteq C_2$. Define $f(x)=\frac{g_2(x)}{g_1(x)}$ over $\mathbb{F}_r[x]/(x^{n}-1)$. Then for any pair of nonnegative integers $c_l,\;c_r$ satisfying $c_l+c_r<\mathrm{ord}\left(f(x)\right)$, there exists a  $(c_l,c_r)$-$[[n+c_l+c_r,2k_2-n]]$ quantum synchronizable code that can correct up to $\left \lfloor \frac{d_1-1}{2} \right \rfloor$ bit errors and $\left \lfloor \frac{d_2-1}{2} \right \rfloor$ phase errors.
\end{lemma}

\subsection{Maximum misalignment tolerance}
\hspace{1,5em}According to Lemma \ref{theorem2}, the order of the polynomial $f(x)$ indicates the tolerable of the quantum synchronizable codes to misalignment errors. To get a good quantum synchronizable code, cyclic codes with good minimum distance are required, while ensuring $\mathrm{ord}\left(f(x)\right)$ to be as large as possible. Since $f(x)=\frac{g_2(x)}{g_1(x)}$, and $g_1(x)|(x^n-1)$, $g_2(x)=(x^n-1)$, then $f(x)|(x^n-1)$. Then the tolerable magnitude of the quantum synchronizable codes constructed is upper bounded by its length $n$.

\begin{lemma}\label{theorem3}\cite{2003FiniteField}
	Let $f(x)\in \mathbb{F}_r[x]$ be an irreducible polynomial over $\mathbb{F}_r$ of degree $\ell_k$ and with $f(0)\neq 0$. Then the order of $f(x)$ is equal to the order of any root of $f(x)$ in the multiplicative group $\mathbb{F}_{r^{\ell_k}} \backslash \{0\}$.
\end{lemma}

By Lemma \ref{lemma3}, it is clear that the size of cyclotomy coset $C_{(t,2q)}$ is $|C_{(t,2q)}|=\ell_k$, where $t\in \{1,2,3,...,q-1,q+1,...,2q-1\}$. Then the degree of irreducible polynomial over $\mathbb{F}_r$ is $\ell_k$. Moreover, Since $q=4m+3$ is an odd prime, and by $(\ref{e10})$, $(\ref{e11})$, the order of any root of $M_t(x)$ is $2q$, where $t\in D_0^{(2q)}\cup D_1^{(2q)}$. Then the order of $M_t(x)$ is $2q$.

\begin{theorem}\label{theorem4}
	Let $q=4m+3$ be an odd prime, $r\in D_0^{(2q)}$, and $\mathrm{ord}_{(2q)}(r)=\ell_k$, $\theta=\frac{q-1}{2\ell_k}$. Then for any nonnegative integers $c_l$ and $c_r$ satisfying $c_l+c_r <2q$, there exists a quantum synchronizable code with parameters $(c_l,c_r)$-$[[2q+c_l+c_r,q+2z\ell_k+1]]_r$, where $z$ is a nonnegative integer such that $0\leq z\leq  \theta-2=\frac{q-4\ell_k-1}{2\ell_k}$.
\end{theorem}
\begin{proof}
	Since $q=4m+3$ is an odd prime, then $|D_i^{(2q)}|=|E_i^{(2q)}|=\frac{q-1}{2}=2m+1$. By Lemma \ref{lemma3}, $|C_{t,2q}|=\ell_k$ for any $t\in \{1,2,...,q-1,q+1,...,2q-1\}$, and $\ell_k |(2m+1)$. Then it follows that dual-containing code $C_{D_i}^{(2q)}$ has augmented codes if and only if $g_{D_i}^{(2q)}(x)$ has at least $\theta=\frac{2m+1}{\ell_k}>1$ irreducible factors over $\mathbb{F}_r$.
	
	According to $(a)$ in Lemma \ref{lemma7}, the dual-containing code $C_{D_i}^{(2q)}$ has an augmented code $C_{D_i}'$ with parameters $[2q,2q-(\frac{q-1}{2}-|A|\ell_k)]_r=[2q,\frac{3q+1}{2}+|A|\ell_k]_r$. Furthermore, taking a set $A'$ such that $A\subset A'\subset \{d_{i_1},d_{i_2},...,d_{i_\theta}\}$, we can get the augmented cyclic code $C_{D_i}''$ of $C_{D_i}'$, and the parameter of $C_{D_i}''$ is $[2q,\frac{3q+1}{2}+|A'|\ell_k]_r$.
	
	Let $|A|=z$. Since $|A|<|A'|<\theta $, then $0\leq z\leq \theta-2=\frac{q-4\ell_k-1}{2\ell_k}$. Furthermore, since $\mathrm{ord}(M_t(x))=2q$ and $t\in D_0^{(2q)}\cup D_1^{(2q)}$, then the synchronization capabilities of quantum synchronizable codes obtained by $C_{D_i}^{(2q)}$ and their augmented cyclic codes reach the upper bound. By Lemma $\ref{theorem2}$ we can obtain a quantum synchronizable codes with parameters $(c_l,c_r)$-$[[2q+c_l+c_r,q+2z\ell_k+1]]_r$, where $c_l+c_r<2q$.
\end{proof}

Moreover, by using cyclic codes $C_{D_iE_j}^{(2q)}$ and their augmented codes, we can also construct a quantum synchronizable code whose synchronization capabilities reach the upper bound.

\begin{theorem}\label{theorem5}
	Let $q=4m+3$ be an odd prime, $r\in D_0^{(2q)}$, and $\mathrm{ord}_{(2q)}(r)=\ell_k$, $\theta=\frac{q-1}{2\ell_k}$. Then for any nonnegative integers $c_l$ and $c_r$ satisfying $c_l+c_r <2q$, there exists a quantum synchronizable code with parameters $(c_l,c_r)$-$[[2q+c_l+c_r,2y\ell_k+2]]_r$, where $y$ is a nonnegative integer such that $0\leq y\leq  \theta-2=\frac{q-4\ell_k-1}{2\ell_k}$.
\end{theorem}
\begin{proof}
	Since $q=4m+3$ is an odd prime,  $|D_i^{(2q)}|=|E_i^{(2q)}|=\frac{q-1}{2}=2m+1$. By Lemma \ref{lemma3}, $|C_{t,2q}|=\ell_k$ for any $t\in \{1,2,...,q-1,q+1,...,2q-1\}$, and $\ell_k |(2m+1)$. Then it follows that the dual-containing code $C_{D_iE_j}^{(2q)}$ has augmented codes if and only if $g_{D_i}^{(2q)}(x)g_{E_j}^{(2q)}(x)$ has at least $\theta = \frac{4m+2}{\ell_k}>1$ irreducible factors over $\mathbb{F}_r$.
	
	According to $(c)$ in Lemma \ref{lemma7}, the dual-containing code $C_{D_iE_j}^{(2q)}$ has an augmented code $C_{D_iE_j}'$ with parameter $[2q,q+|S|\ell_k+1]_r$. Moreover, taking a set $S'$ such that $S\subset S'\subset \{d_{i_1},d_{i_2},...,d_{i_\theta} \}\cup \{e_{i_1},e_{i_2},...,e_{i_\theta} \}$, we obtain the augmented cyclic codes $C_{D_iE_j}''$ of $C_{D_iE_j}'$, and $C_{D_iE_j}''=[2q,q+|S'|\ell_k+1]_r$.
	
	Let $|S|=y$. Since $|S|<|S'|<2\theta$, then $0\leq y\leq 2\theta-2=\frac{q-2\ell_k-1}{\ell_k}$. Note that $\mathrm{ord}(M_t(x))=2q$ and $t\in\{D_0^{(2q)}\cup D_1^{(2q)}\}$, then if the quotient of the generator polynomial of $C_{D_iE_j}'$ and $C_{D_iE_j}''$ is product of some $M_t(x),\;t\in D_0^{(2q)}\cup D_1^{(2q)}$. Hence we obtain a quantum synchronizable code with parameters $(c_l,c_r)$-$[[2q+c_l+c_r,2y\ell_k+2]]_r$ whose synchronization capabilities reach the upper bound.
\end{proof}

\begin{example}
	Let $q=4m+3=19$, and $r=11\in D_0^{(38)}$. In this example, we only consider the construction of quantum synchronizable codes from the cyclic code $C_{D_0E_0}^{(38)}=\langle g_{D_0}^{(38)}(x)g_{E_0}^{(38)}(x)\rangle$ and its augmented cyclic codes. Since $\ell_k=\mathrm{ord}_{38}(11)=3$, then $\theta = \frac{q-1}{2\ell_k}=3$. By Theorem \ref{theorem5}, $0\leq z \leq 4$.
	
	Let $z=4$, and $S=\{2,5,9,10\}$. Then the dual-containing code $C_{D_0E_0}'=\langle \frac{g_{D_0}^{(38)}(x)g_{E_0}^{(38)}(x)}{M_2(x)M_5(x)M_9(x)M_{10}(x)} \rangle$ has parameter $[38,32,5]_{11}$. Let $S'=\{1,2,5,9,10\}$. The code $C_{D_0E_0}''=\langle \frac{g_{D_0}^{(38)}(x)g_{E_0}^{(38)}(x)}{M_1(x)M_2(x)M_5(x)M_9(x)M_{10}(x)}\rangle$ is the augmented code of $C_{D_0E_0}'$. Moreover, $C_{D_0E_0}''$ has parameter $[38,35,2]_{11}$. Note that the quotient of the generator polynomial of $C_{D_0E_0}'$ and $C_{D_0E_0}''$ is $M_1(x)$, and $\mathrm{ord}(M_1(x))=38$ due to $1 \in D_0^{(38)}$. By Lemma \ref{theorem2}, and from ${C_{D_0E_0}'}^{\perp} \subset C_{D_0E_0}'\subset C_{D_0E_0}''$, we obtain a quantum synchronizable code with parameters $(c_l,c_r)$-$[[38+c_l+c_r,26]]_{11}$ and $c_l+c_r<38$, whose error-correcting and synchronization capabilities reach the upper bound.
\end{example}

\section{Conclusion}
\hspace{1.5em}In this paper, we constructed two classes of quantum synchronizable codes by using cyclotimic classes of order two over $\mathbb{Z}_{2q}$. In certain conditions, the synchronization capabilities of the codes reach the upper bound. Moreover, by choosing the optimal or almost optimal dual-containing codes, the quantum synchronizable codes also posses good error-correcting capability towards bit errors and phase errors. Finally, we constructed a quantum synchronizable code with parameters $(c_l,c_r)$-$[[38+c_l+c_r,26]]_{11}$ and $c_l+c_r<38$, whose synchronizable capability reach the upper bound. For future work, more quantum synchronizable codes from generalized cyclotomic classes over $\mathbb{Z}_n$ are worthy studying.

\bibliography{QSCs_From_Cyclotomic_Classes_of_Order_Two_over_Z2q}
\end{document}